\documentclass[letterpaper,10pt,conference]{ieeeconf}

\IEEEoverridecommandlockouts
\usepackage{cite}
\usepackage{amsmath,amssymb,amsfonts,mathtools}
\usepackage{graphicx}
\usepackage{textcomp}
\usepackage{comment}
\usepackage{xcolor}

\usepackage{algorithm}
\usepackage{algpseudocode}

\usepackage{amsthm}

\newtheorem{theorem}{Theorem}
\newtheorem{corollary}[theorem]{Corollary}
\newtheorem{definition}{Definition}
\newtheorem{lemma}[theorem]{Lemma}
\newtheorem{proposition}[theorem]{Proposition}

\newtheorem{problem}{Problem}

\newcommand{\longthmtitle}[1]{\mbox{}\emph{(#1):}}

\usepackage[hidelinks]{hyperref}

\begin{document}
\bstctlcite{BSTcontrol}

\title{\LARGE \bf Functional Observers for Output-Feedback Safety \\with High-Order Control Barrier Functions}

\author{Shima Sadat Mousavi, Pol Mestres,
and Aaron D. Ames %%
\thanks{The authors are with the Department of Mechanical and Civil
Engineering, California Institute of Technology, Pasadena, CA 91125, USA.
Emails: \texttt{smousavi,mestres,ames@caltech.edu}.
This research is supported by The Boeing Company.}
}

\maketitle

\begin{abstract}
This paper studies output-feedback safety filtering for linear systems using high-order control barrier functions (HOCBFs) when the full state is unavailable. Since the HOCBF conditions depend only on specific state functionals, full-state reconstruction is unnecessary; instead, their implementation requires only functional observability or detectability of these quantities. We show that the HOCBF gains determine the required functionals and hence their observability properties, coupling safety and estimation design. We exploit this structure using scalar functional observers and develop a polynomial co-design framework that characterizes HOCBF observability and detectability, observer existence and order, and convergence-rate limitations. A certified one-sided estimation-error bound then enables safe output-feedback
implementation.
\end{abstract}

\section{Introduction}

Safety-critical control enforces state constraints while preserving a
nominal control objective. Control Barrier Functions (CBFs) express
safety through forward invariance and incorporate constraints into
optimization-based controllers
\cite{Ames-Xu-Grizzle-Tabuada:17}.
Exponential and High-Order CBFs (HOCBFs) extend this framework to
higher-relative-degree constraints
\cite{Nguyen-Sreenath:16,Xiao-Belta:19}. Their implementation typically
requires state-dependent quantities, often motivating a full-state
observer when only outputs are measured. Full-state reconstruction,
however, may be unnecessary or even impossible for unobservable
systems. This raises the central question: what state information is
actually required by an HOCBF safety filter, and how should the HOCBF
gains (i.e., the class $\mathcal{K}$ functions associated with the HOCBF construction) be chosen to make that
information estimable?

Several approaches extend CBF-based safety to partial or uncertain
state information. Measurement-robust CBFs account for state-estimation
or perception errors, including high-relative-degree constraints
\cite{Dean-Taylor-Cosner-Recht-Ames:21,
Cosner-Singletary-Taylor-Molnar-Bouman-Ames:21,
Oruganti-Naghizadeh-Ahmed:23}, while observer-based methods incorporate
estimation-error bounds into barrier constraints
\cite{Wang-Xu:22,Agrawal-Panagou:23}. %Related work considers robust output CBFs, joint observer--barrier synthesis, and backup CBFs under state uncertainty and bounded inputs \cite{Lindemann-Robey-Jiang-Das-Tu-Matni:24, Jagabathula-Basu-Jagtap:26, vanWijk-Molnar-Coogan-Majji-Ames-Burdick:26}.
Related work considers robust output CBFs and backup CBFs under
state uncertainty and bounded inputs \cite{Lindemann-Robey-Jiang-Das-Tu-Matni:24,
vanWijk-Molnar-Coogan-Majji-Ames-Burdick:26}.
Most closely related, \cite{Menner-Hussain-Lavretsky:26} studies
output-feedback control of LTI systems with operational constraints
using a full-state Luenberger observer and provides a simple condition
ensuring that safety is maintained after a transient despite
state-estimation error. These methods typically reconstruct the state
and then account for estimation error, rather than characterize the
weaker requirement of estimating only the functionals needed by the
HOCBF inequalities. Thus, they do not identify when output-feedback
safety remains possible without full-state reconstruction.

%Functional observers estimate prescribed state functionals and may exist even when the full system is unobservable. Their existence, functional observability and detectability, reduced-order construction, and assignable error dynamics are well studied \cite{Luenberger:66,MD:00,Fernando-Trinh-Jennings:10, Rotella-Zambettakis:16,KV:21,MD-TF:23,MD-TF:25}. In this literature, however, the functional is fixed. In an HOCBF safety filter, it depends on the HOCBF gains, which are design parameters. Gain selection therefore affects both the barrier condition and whether the required functional is observable or detectable, raising the question of how to select the gains to ensure estimability with fast convergence while maintaining a large admissible set.

Functional observers estimate prescribed state functionals and may exist
even when the full system is unobservable. Their existence, functional
observability and detectability, reduced-order construction, and
assignable error dynamics are well studied
\cite{Luenberger:66,MD:00,Fernando-Trinh-Jennings:10,
Rotella-Zambettakis:16,KV:21,MD-TF:23,MD-TF:25}.
In this literature, however, the functional is fixed, whereas in an
HOCBF safety filter it depends on the HOCBF gains. Gain selection
therefore affects both the barrier condition and whether the required
functional is observable or detectable, motivating gains that ensure
estimability with fast convergence while maintaining a large
admissible set.

We study this problem for linear systems with affine state constraints.
Recent work has characterized feasibility, compatibility, explicit
implementation, and closed-loop dynamics of safety filters in this
setting
\cite{Mousavi-Tan-Ames:25,Mousavi-Cohen-Mestres-Ames:26,
Mestres-Mousavi-Ames:26}, but output-feedback design remains largely
unexplored. HOCBF theory for general control-affine systems has also
studied feasibility and conservatism
\cite{Xiao-Belta:19,Zhang-Zhai-Xiong-Xia:24}, but not output feedback
from a functional-observability perspective or how the HOCBF gains
affect observer existence, order, and convergence. This coupling
between HOCBF design and functional-observer dynamics therefore
remains uncharacterized.

%This work bridges this gap for LTI systems with affine safety constraints. First, we construct an output-feedback safety filter that estimates only the scalar functionals required by the HOCBF constraints and uses a certified one-sided error bound to guarantee safety. Second, we co-design the HOCBF gains and functional observers, characterizing HOCBF observability and detectability for fixed gains, when these properties are achievable by gain selection, and the resulting observer-order and convergence-rate limits, even when the full state is unobservable. Third, after selecting the gain values and observer, we exploit the remaining gain-ordering freedom to enlarge the HOCBF admissible set without changing the final HOCBF inequality or observer design.

We bridge this gap for LTI systems with affine safety constraints.
First, we construct an output-feedback safety filter that estimates only
the scalar functionals required by the HOCBF constraints and uses a
certified one-sided error bound to guarantee safety. Second, we co-design
the HOCBF gains and functional observers, characterizing fixed-gain
HOCBF observability and detectability, their achievability by gain
selection, and the resulting observer-order and convergence-rate limits,
even when the full state is unobservable. Third, after selecting the
gains and observer, we exploit gain-ordering freedom to enlarge the
HOCBF admissible set without changing the final HOCBF inequality or
observer design.

\section{Problem Formulation}
\label{sec:problem-formulation}

This section formulates the output-feedback safety problem and
introduces the functional-observer framework.

\emph{Notation:}
Let $\mathbb R$ and $\mathbb R_{\geq0}$ denote the reals and nonnegative
reals, and $I_k$ the $k\times k$ identity. For a matrix $M$,
$\ker(M)$, $\operatorname{row}(M)$, and $M_{i,:}$ denote its kernel,
row space, and $i$th row; for a subspace $\mathcal V$,
$\mathcal V^\perp$ denotes its orthogonal complement. For polynomials
$p,q$, $p\mid q$, $\gcd(p,q)$, and $\deg(p)$ denote divisibility,
the monic greatest common divisor, and degree (counting multiplicities). The Euclidean norm is $\|\cdot\|$; vector inequalities are componentwise.

\subsection{High-Order Affine Control Barrier Functions }
\label{subsec:hocbf-filter}

Consider the linear time-invariant system
\begin{equation}
    \dot x=Ax+Bu,
    \qquad
    y=Cx,
    \label{eq:system}
\end{equation}
where $x\in\mathbb R^n$, $u\in\mathbb R^m$, and
$y\in\mathbb R^q$. Only the output $y$ is available to the
controller.

Safety is specified by the affine constraints
\begin{equation*}
    h_i(x):=f_i^\top x+d_i\geq0,
    \qquad i=1,\ldots,p,
    %\label{eq:affine-constraints}
\end{equation*}
with safe set
\(
    \mathcal S
    :=
    \bigcap_{i=1}^{p}
    \{x\in\mathbb R^n:h_i(x)\geq0\}.
\)
Assume that $h_i$ has relative degree $r_i$, that is,
\[
    f_i^\top A^kB=0,
    \quad k=0,\ldots,r_i-2,
    \qquad
    f_i^\top A^{r_i-1}B\neq0,
\]
where the first condition is omitted when $r_i=1$.

For each constraint, define the ordered gain vector
\[
    \boldsymbol{\alpha}_i
    :=
    (\alpha_{i1},\ldots,\alpha_{ir_i}) \in \mathbb R^{r_i},
    \qquad
    \alpha_{ij}>0,
\]
and the associated HOCBF polynomial
\begin{equation*}
    \phi_i(s)
    :=
    \prod_{j=1}^{r_i}(s+\alpha_{ij}).
    %\label{eq:hocbf-polynomial}
\end{equation*}
The polynomial $\phi_i$ depends only on the selected gain values,
whereas $\boldsymbol{\alpha}_i$ also specifies their order in the
HOCBF recursion
\begin{equation}
    \psi_{i,0}:=h_i,
    \qquad
    \psi_{i,j}
    :=
    \dot\psi_{i,j-1}
    +\alpha_{ij}\psi_{i,j-1},
    \quad j=1,\ldots,r_i,
    \label{eq:hocbf-recursion}
\end{equation}
where the derivatives are evaluated along \eqref{eq:system}. The
functions up to order $r_i-1$ define the HOCBF admissible sets
\begin{equation}
    \mathcal C_{i,\boldsymbol{\alpha}_i}
        :=
        \bigcap_{j=0}^{r_i-1}
        \{x:\psi_{i,j}(x)\geq0\},
    \qquad
    \mathcal C_{\boldsymbol{\alpha}}
        :=
        \bigcap_{i=1}^{p}
        \mathcal C_{i,\boldsymbol{\alpha}_i},
    \label{eq:hocbf-set}
\end{equation}
where
$\boldsymbol{\alpha}
:=(\boldsymbol{\alpha}_1,\ldots,\boldsymbol{\alpha}_p)$.
Since $\psi_{i,0}=h_i$, one has
$\mathcal C_{\boldsymbol{\alpha}}\subseteq\mathcal S$.

For system \eqref{eq:system}, the  HOCBF conditions
$\psi_{i,r_i}\geq0$ yield
\begin{equation}
    L_\phi x+Du+b_\phi\geq0,
    \label{eq:hocbf-inequality}
\end{equation}
where $L_\phi\in\mathbb R^{p\times n}$,
$D\in\mathbb R^{p\times m}$, and $b_\phi\in\mathbb R^p$ are defined
rowwise by
\begin{equation*}
\begin{aligned}
    [L_\phi]_{i,:}
        &:=f_i^\top\phi_i(A),&
    [D]_{i,:}
        &:=f_i^\top A^{r_i-1}B,&
    [b_\phi]_i
        &:=d_i\phi_i(0).
\end{aligned}
%\label{eq:hocbf-matrices}
\end{equation*}
By \cite{Xiao-Belta:19}, any locally Lipschitz input satisfying
\eqref{eq:hocbf-inequality} renders
$\mathcal C_{\boldsymbol{\alpha}}$ forward invariant.

For each $i=1,\ldots,p$, define the scalar HOCBF functional
\begin{equation*}
    \ell_i:=f_i^\top\phi_i(A),
    \qquad
    z_i:=\ell_i x,
    %\label{eq:scalar-hocbf-functional}
\end{equation*}
and collect these quantities as
\[
    z:=(z_1,\ldots,z_p)^\top=L_\phi x.
\]
Hence, \eqref{eq:hocbf-inequality} becomes
\(
    z+Du+b_\phi\geq0,
\)
so the HOCBF constraints require only $z$, rather than the full state $x$.

For a nominal output-feedback controller $u_{\rm nom}(y,t)$, consider
the safety filter
\begin{equation}
\begin{aligned}
    u_{\rm sf}(y,z,t)
    :=\arg\min_{u\in\mathbb R^m}\quad&
        \frac12\|u-u_{\rm nom}(y,t)\|^2\\
    \mathrm{s.t.}\quad&
        z+Du+b_\phi\geq0.
\end{aligned}
\label{eq:ideal-safety-filter}
\end{equation}
When feasible, the strictly convex objective gives a unique minimizer.
Feasibility is assumed throughout; full row rank of $D$ is sufficient
for arbitrary right-hand sides.

The gain values $\boldsymbol{\alpha}$ and their ordering play different roles.
The polynomial $\phi_i$ determines $\ell_i$ and $[b_\phi]_i$, whereas $D$ is gain-independent. Permuting the same gain values $\boldsymbol{\alpha}_i$ therefore leaves the final HOCBF inequality and \eqref{eq:ideal-safety-filter} unchanged, but may change $\mathcal C_{i,\boldsymbol{\alpha}_i}$.

\subsection{Functional-Observer Preliminaries}
\label{subsec:functional-observer-preliminaries}

This subsection summarizes the functional-observer preliminaries
needed for the co-design analysis.

For $M\in\mathbb R^{s\times n}$ and $k\geq0$, define
\[
    \mathcal O_k(M,A)
    :=
    \begin{bmatrix}
        M^\top & (MA)^\top & \cdots & (MA^k)^\top
    \end{bmatrix}^{\!\top}.
\]
The corresponding unobservable subspace is
$\mathcal U(M):=\ker\mathcal O_{n-1}(M,A)$.
In particular, $\mathcal U(C)$ is the output-unobservable subspace:
two initial states driven by the same input generate the same output
trajectory if and only if their initial-state difference lies in
$\mathcal U(C)$.

\begin{definition}\longthmtitle{Functional observability and detectability}
\label{def:functional-properties}
A scalar functional $z=\bar\ell x$ is \emph{functionally observable}
from $y=Cx$ if any two initial states that produce the same output
trajectory under the same input also produce the same $z$ trajectory.
It is \emph{functionally detectable} if the difference between their
$z$ trajectories converges to zero
\cite{MD-TF:23,MD-TF:25}.
\end{definition}

For system \eqref{eq:system}, equivalently,
\begin{align*}
    z\text{ functionally observable}
        &\iff
        \mathcal U(C)\subseteq\mathcal U(\bar\ell),\\
    z\text{ functionally detectable}
        &\iff
        \lim_{t\to\infty}\bar\ell e^{At}v=0,
        \; \forall v\in\mathcal U(C).
\end{align*}
Functional observability implies functional detectability, but not
conversely. Both properties may hold even when $(A,C)$ is not
observable.

We now specialize these notions to HOCBFs, distinguishing fixed gains
from gain selection.

\begin{definition}\longthmtitle{HOCBF observability and detectability for fixed gains}
\label{def:fixed-hocbf-observability}
For a fixed gain vector
$\boldsymbol{\alpha}_i=(\alpha_{i1},\ldots,\alpha_{ir_i})$ with
$\alpha_{ij}>0$, the pair $(h_i,\boldsymbol{\alpha}_i)$ is
\emph{HOCBF-observable} (resp., \emph{HOCBF-detectable}) from
$y$ if the corresponding HOCBF functional
\(
    z_i=\ell_i x=f_i^\top\phi_i(A)x
\)
is functionally observable (resp., functionally detectable)
from $y$.
\end{definition}

\begin{definition}\longthmtitle{HOCBF-observable and detectable constraints}
\label{def:hocbf-observable-constraint}
The constraint $h_i$ is \emph{HOCBF-observable} (resp.,
\emph{HOCBF-detectable}) from $y$ if there exists a positive gain
vector $\boldsymbol{\alpha}_i$ such that
$(h_i,\boldsymbol{\alpha}_i)$ is HOCBF-observable (resp.,
HOCBF-detectable) from $y$.
\end{definition}

We also define
\begin{equation}
    \kappa_C
    :=
    \min\left\{
        k\geq0:
        \operatorname{row}\mathcal O_k(C,A)
        =
        \mathcal U(C)^\perp
    \right\}.
    \label{eq:output-saturation-order}
\end{equation}
Hence, $\kappa_C$ is the smallest $k$ for which
$\mathcal O_k(C,A)$ captures all state information available from the
output trajectory.

To estimate $z_i$ directly, consider the scalar functional observer
\cite{MD:00,MD-TF:23,MD-TF:25}
\begin{equation}
\begin{aligned}
    \dot\eta_i &=N_i\eta_i+J_i y+G_i u,\\
    \hat z_i &=R_i\eta_i+E_i y,
\end{aligned}
\label{eq:scalar-functional-observer}
\end{equation}
where $\eta_i\in\mathbb R^{\nu_i}$ and $\nu_i$ is the observer order.
We consider standard realizations with autonomous linear error dynamics,
for which there exists $T_i\in\mathbb R^{\nu_i\times n}$ satisfying
\begin{equation}
    T_iA=N_iT_i+J_iC,\quad
    G_i=T_iB,\quad
    \ell_i=R_iT_i+E_iC.
    \label{eq:observer-embedding}
\end{equation}
%Given $\nu_i$ and $N_i$, the remaining matrices are obtained by solving \eqref{eq:observer-embedding}; see \cite{MD:00,Rotella-Zambettakis:16,KV:21}.
Given $\nu_i$, one may choose an observable canonical realization
$(N_i,R_i)$, e.g., with $R_i=[\,1\ 0\ \cdots\ 0\,]$, and obtain the
remaining matrices from \eqref{eq:observer-embedding}; see
\cite{MD:00,Rotella-Zambettakis:16,KV:21}.

\begin{comment}
To estimate $z_i$ directly, consider the scalar functional observer
\cite{MD:00,MD-TF:23,MD-TF:25}
\begin{equation}
\begin{aligned}
    \dot\eta_i &=N_i\eta_i+J_i y+G_i u,\\
    \hat z_i &=R_i\eta_i+E_i y,
\end{aligned}
\label{eq:scalar-functional-observer}
\end{equation}
where $\eta_i\in\mathbb R^{\nu_i}$ and $\nu_i$ is the observer order.
The observer matrices can be obtained using standard
functional-observer synthesis methods \cite{MD:00,KV:21}.

We consider functional observers whose estimation errors admit
autonomous linear dynamics. Specifically, suppose there exists
$T_i\in\mathbb R^{\nu_i\times n}$ satisfying
\begin{equation}
    T_iA=N_iT_i+J_iC,
    \quad
    G_i=T_iB,
    \quad
    \ell_i=R_iT_i+E_iC,
    \label{eq:observer-embedding}
\end{equation}
as in standard functional-observer constructions
\cite{Rotella-Zambettakis:16}.
\end{comment}

Define the observer error state and estimation error as
\[
    \varepsilon_i:=T_ix-\eta_i,
    \qquad
    e_i:=z_i-\hat z_i.
\]
Using the system and observer dynamics,
\(
    \dot\varepsilon_i
    =
    N_i\varepsilon_i
    +(T_iA-N_iT_i-J_iC)x
    +(T_iB-G_i)u.
\)
Hence, \eqref{eq:observer-embedding} gives
\begin{equation}
    \dot\varepsilon_i=N_i\varepsilon_i,
    \qquad
    e_i=R_i\varepsilon_i.
    \label{eq:scalar-observer-error}
\end{equation}
Because $u$ is available to the observer, its contribution can be
exactly reproduced by choosing $G_i=T_iB$; thus, this condition places no additional restriction on $B$.

We take the error realization to be minimal and define
\begin{equation}\label{eq:beta}
    \beta_i(s):=\det(sI_{\nu_i}-N_i),
\end{equation}
with $\beta_i=1$ when $\nu_i=0$.

The following existence condition is standard
\cite{KV:21}.

\begin{proposition}\longthmtitle{Functional-observer existence}
\label{prop:functional-observer-existence}
Fix $\nu_i$ and a real monic polynomial $\beta_i$ of degree $\nu_i$.
An observer satisfying \eqref{eq:observer-embedding} with error
polynomial $\beta_i$ exists if and only if
\[
    \ell_i\beta_i(A)
    \in
    \operatorname{row}\mathcal O_{\nu_i}(C,A).
\]
For the minimal realization \eqref{eq:scalar-observer-error},
$e_i(t)\to0$ for all initial errors if and only if $\beta_i$ is
Hurwitz.
\end{proposition}

\subsection{Output-Feedback Safety Design Problem}
\label{subsec:functional-estimation}

The filter \eqref{eq:ideal-safety-filter} requires the HOCBF
functionals $z_i$. Since $\ell_i=f_i^\top\phi_i(A)$, these functionals
depend on the HOCBF gains, coupling HOCBF and observer design.

\begin{problem}
\label{prob:main}
For system \eqref{eq:system}, implement
\eqref{eq:ideal-safety-filter} from output measurements using the
functional observers \eqref{eq:scalar-functional-observer} while
guaranteeing safety. Characterize HOCBF observability and detectability, observer
existence, order, and rate under gain selection, and determine the
ordering maximizing $\mathcal C_{\boldsymbol{\alpha}}$. 
\end{problem}

\section{Safe Functional-Observer Filtering}
\label{sec:safe-filtering}

This section derives a certified error bound for output-feedback HOCBF safety filtering.

Let
\[
    \hat z:=(\hat z_1,\ldots,\hat z_p)^\top,
    \qquad
    e:=z-\hat z.
\]
If $e_i<0$, then $\hat z_i$ overestimates the true HOCBF functional
$z_i$ and may compromise safety, whereas $e_i>0$ is conservative.
Hence, only a lower error bound is needed. Suppose a known function
$\rho:\mathbb R_{\geq0}\rightarrow\mathbb R_{\geq0}^{p}$ satisfies
\begin{equation}
    e(t)\geq-\rho(t),
    \qquad t\geq0.
    \label{eq:error-lower-bound}
\end{equation}

Using this bound, replace the safety filter
\eqref{eq:ideal-safety-filter} by the implementable output-feedback
filter
\begin{equation}
\begin{aligned}
    \widetilde u(y,\hat z,t)
    :=\arg\min_{u\in\mathbb R^m}\quad&
        \frac12\|u-u_{\rm nom}(y,t)\|^2\\
    \mathrm{s.t.}\quad&
        \hat z+Du+b_\phi\geq\rho(t).
\end{aligned}
\label{eq:output-feedback-filter}
\end{equation}
Under the feasibility assumption of Section~\ref{subsec:hocbf-filter}, the strictly convex objective gives a unique minimizer. The next result shows that the tightening by $\rho$ preserves safety.

\begin{proposition}\longthmtitle{Safety under a one-sided error bound}
\label{prop:one-sided-safety}
Suppose that \eqref{eq:error-lower-bound} holds and the solution of
\eqref{eq:output-feedback-filter} is locally Lipschitz. Then
$\mathcal C_{\boldsymbol{\alpha}}$ is forward invariant.
\end{proposition}

\begin{proof}
Since $e=z-\hat z$, the constraint in
\eqref{eq:output-feedback-filter} and
\eqref{eq:error-lower-bound} give
\[
    z+Du+b_\phi
    =
    \hat z+Du+b_\phi+e
    \geq
    \rho-\rho
    =0.
\]
Thus, $\widetilde u$ satisfies the true HOCBF inequality
\eqref{eq:hocbf-inequality}. Forward invariance follows from the
standard HOCBF result \cite{Xiao-Belta:19}.
\end{proof}

Under \eqref{eq:error-lower-bound}, $\rho_i(t)$ is minimal: any smaller
margin could allow $e_i(t)=-\rho_i(t)$ to violate the true HOCBF
inequality.

We next construct such a bound for the autonomous error dynamics in
Section~\ref{subsec:functional-observer-preliminaries}. Define
\(
    \chi
    :=
    \begin{bmatrix}
        \varepsilon_1^\top & \cdots & \varepsilon_p^\top
    \end{bmatrix}^{\!\top},
    \) 
    \(A_e:=\operatorname{blkdiag}(N_1,\ldots,N_p),
\)
and
\(
    C_e:=\operatorname{blkdiag}(R_1,\ldots,R_p).
\)
Then
\begin{equation}
    \dot\chi=A_e\chi,
    \qquad
    e=C_e\chi.
    \label{eq:error-realization}
\end{equation} 
Assume
\begin{equation}
    \chi(0)^\top P\chi(0)\leq\bar V_0,
    \label{eq:initial-error-bound}
\end{equation}
where $P\succ0$ and $\bar V_0\geq0$ are known.

For each $i=1,\ldots,p$, define
\begin{equation}
    q_i(t):=[C_e]_{i,:}e^{A_et},
    \quad
    \rho_i(t):=
    \sqrt{\bar V_0\,q_i(t)P^{-1}q_i(t)^\top},
    \label{eq:componentwise-error-margin}
\end{equation}
and let
\(
    \rho:=(\rho_1,\ldots,\rho_p)^\top.
\)
The following result certifies the required error bound.

\begin{comment}\begin{proposition}[Componentwise estimation-error bound]
\label{prop:componentwise-error-bound}
Under \eqref{eq:error-realization} and
\eqref{eq:initial-error-bound}, $\rho$ satisfies
\eqref{eq:error-lower-bound}. If $A_e$ is Hurwitz, then
$\rho(t)$ converges to zero exponentially.
\end{proposition}
\end{comment}
\begin{proposition}\longthmtitle{Componentwise estimation-error bound}
\label{prop:componentwise-error-bound}
Under \eqref{eq:error-realization} and
\eqref{eq:initial-error-bound}, $\rho$ satisfies
\eqref{eq:error-lower-bound}. If $A_e$ is Hurwitz, then
$\rho(t)\to0$ exponentially. Moreover, if $D$ has full row rank,
\[
    \widetilde u(y,\hat z,t)-u_{\rm sf}(y,z,t)\to0.
\]
\end{proposition}

\begin{proof}
From \eqref{eq:error-realization},
\(
    e_i(t)=q_i(t)\chi(0).
\)
By weighted Cauchy--Schwarz and
\eqref{eq:initial-error-bound},
\(
    |e_i(t)|
    \leq
    \sqrt{\bar V_0\,q_i(t)P^{-1}q_i(t)^\top}
    =\rho_i(t).
\)
Hence, $e_i(t)\geq-\rho_i(t)$ for all $i$, proving
\eqref{eq:error-lower-bound}. For fixed $i,t$ with $q_i(t)\neq0$,
the lower bound is attained by an initial error on the boundary of
\eqref{eq:initial-error-bound} in the direction
$-P^{-1}q_i(t)^\top$, so it is pointwise tight.

If $A_e$ is Hurwitz, then $e(t)$ and $\rho(t)$ converge to zero
exponentially. Thus, $\hat z(t)\to z(t)$ and the constraint in
\eqref{eq:output-feedback-filter} converges to that in
\eqref{eq:ideal-safety-filter}. Under the full-row-rank assumption on
$D$, continuity of the unique QP solution then gives
\(
    \widetilde u(y,\hat z,t)-u_{\rm sf}(y,z,t)\to0 .
\)
\end{proof}

\section{HOCBF Gain--Observer Co-Design}
\label{sec:co-design}

This section co-designs HOCBF gains and functional observers through
unobservable-factor allocation, yielding observability, detectability,
order, and rate results.

\subsection{Unobservable-Factor Allocation}
\label{subsec:mode-allocation}

%This subsection identifies the output-unobservable dynamics relevant to each safety constraint and determines which factors can be removed by the HOCBF polynomial and which must appear in the observer-error polynomial.

This subsection identifies the relevant unobservable dynamics and their
allocation between the HOCBF and observer-error polynomials.

Consider two trajectories $x_1,x_2$ of \eqref{eq:system} with identical
input and output, and let $\delta x:=x_1-x_2\in\mathcal U(C)$. Let
$n_u:=\dim\mathcal U(C)$ and
$V_{\mathcal U}\in\mathbb R^{n\times n_u}$ span $\mathcal U(C)$.
Since $\mathcal U(C)$ is $A$-invariant, there exists
$A_{\mathcal U}\in\mathbb R^{n_u\times n_u}$ such that
%Consider two trajectories $x_1$ and $x_2$ of \eqref{eq:system} driven by the same input and producing the same output trajectory. Let
%\(
 %   \delta x:=x_1-x_2.
%\)
%Then $\delta x(t)\in\mathcal U(C)$. Let $n_u:=\dim\mathcal U(C)$ and let $V_{\mathcal U}\in\mathbb R^{n\times n_u}$ have columns forming a basis for $\mathcal U(C)$. Since $\mathcal U(C)$ is $A$-invariant, there exists $A_{\mathcal U}\in\mathbb R^{n_u\times n_u}$ with
\begin{equation}
    AV_{\mathcal U}=V_{\mathcal U}A_{\mathcal U}.
    \label{eq:unobservable-dynamics}
\end{equation}
Writing
\(
    \delta x=V_{\mathcal U}\xi
\)
gives
\(
    \dot\xi=A_{\mathcal U}\xi.
\)

For constraint $i$, define
\[
    \delta h_i
    :=
    h_i(x_1)-h_i(x_2)
    =
    f_i^\top V_{\mathcal U}\xi,
    \qquad
    c_i:=f_i^\top V_{\mathcal U}.
\]
Thus, $c_i$ determines which output-unobservable dynamics affect the
$i$th constraint.

We define the \emph{relevant unobservable polynomial} $\mu_i$ as
follows. If $c_i=0$, set $\mu_i=1$. Otherwise, let $\sigma_i$ be the
smallest positive integer such that
\[
    c_iA_{\mathcal U}^{\sigma_i}
    \in
    \operatorname{span}
    \{c_i,c_iA_{\mathcal U},\ldots,
      c_iA_{\mathcal U}^{\sigma_i-1}\}.
\]
Equivalently, for some $a_{i,0},\ldots,a_{i,\sigma_i-1}$,
\begin{equation}
    c_iA_{\mathcal U}^{\sigma_i}
    +
    \sum_{k=0}^{\sigma_i-1}
    a_{i,k}c_iA_{\mathcal U}^{k}
    =0.
    \label{eq:krylov-dependence}
\end{equation}
Define
\begin{equation}
    \mu_i(s)
    :=
    s^{\sigma_i}
    +
    \sum_{k=0}^{\sigma_i-1}a_{i,k}s^k.
    \label{eq:relevant-unobservable-polynomial}
\end{equation}
Thus, $\mu_i$ is obtained from the first linear dependence in the
Krylov sequence
$c_i,c_iA_{\mathcal U},c_iA_{\mathcal U}^2,\ldots$
\cite{Dumas-Pernet-Wan:05}. Equivalently, it is the monic polynomial
of smallest degree satisfying
\[
    f_i^\top V_{\mathcal U}\mu_i(A_{\mathcal U})=0.
\]
%Hence, $\mu_i$ captures exactly the output-unobservable dynamics relevant to the $i$th constraint. If $A_{\mathcal U}$ is diagonalizable, $\mu_i$ is the product of factors $s-\lambda$ associated with the distinct relevant unobservable modes; otherwise, multiplicities account for generalized-mode chains. Moreover, $\mu_i$ is independent of the basis $V_{\mathcal U}$.
%Hence, $\mu_i$ captures the output-unobservable dynamics relevant to constraint $i$. If $A_{\mathcal U}$ is diagonalizable, $\mu_i$ is the product of $s-\lambda$ over the distinct relevant modes; otherwise, multiplicities encode generalized-mode chains. It is independent of the basis $V_{\mathcal U}$.
Hence, $\mu_i$ captures the unobservable dynamics relevant to constraint
$i$. If $A_{\mathcal U}$ is diagonalizable, it is the product of
$s-\lambda$ over distinct relevant modes; otherwise, multiplicities
encode generalized chains. It is basis-independent.

Recall that $\ell_i=f_i^\top\phi_i(A)$. For the two trajectories,
\[
    \delta z_i
    :=
    z_i(x_1)-z_i(x_2)
    =
    f_i^\top V_{\mathcal U}\phi_i(A_{\mathcal U})\xi.
\]
Thus, common factors of $\mu_i$ and $\phi_i$, including
multiplicities, do not contribute to the output-unobservable part of
$z_i$. Define the remaining unobservable polynomial
\[
    \mu_i^{\rm rem}
    :=
    \frac{\mu_i}{\gcd(\mu_i,\phi_i)}.
\]

The following lemma characterizes $\mu_i$ and $\mu_i^{\rm rem}$.

\begin{lemma}\longthmtitle{Unobservable polynomial characterization}
\label{lem:relevant-unobservable-polynomial}
For any polynomials $p$ and $q$,
\[
    f_i^\top V_{\mathcal U}p(A_{\mathcal U})=0
    \quad\Longleftrightarrow\quad
    \mu_i\mid p,
\]
and
\begin{equation}
    f_i^\top V_{\mathcal U}
    \phi_i(A_{\mathcal U})q(A_{\mathcal U})=0
    \quad\Longleftrightarrow\quad
    \mu_i^{\rm rem}\mid q.
    \label{eq:residual-polynomial-property}
\end{equation}
\end{lemma}

\begin{proof}
For the first equivalence, $\mu_i\mid p$ implies the result by the
definition of $\mu_i$. Conversely, write
$p=q_1\mu_i+r$ with $\deg(r)<\deg(\mu_i)$. Then
$f_i^\top V_{\mathcal U}r(A_{\mathcal U})=0$, so minimality of
$\mu_i$ gives $r=0$.
For the second, apply the first equivalence to $p=\phi_i q$:
\[
    \mu_i\mid\phi_i q
    \quad\Longleftrightarrow\quad
    \mu_i^{\rm rem}\mid q.
\]
\end{proof}

Thus, $\mu_i^{\rm rem}$ captures the remaining unobservable dynamics in $z_i$.
We can now state the main polynomial co-design result.

\begin{theorem}\longthmtitle{Unobservable-factor allocation}
\label{thm:factor-allocation}
Fix $\boldsymbol{\alpha}_i$, and let $\beta_i$ be the observer-error
polynomial defined in \eqref{eq:beta}. Every scalar functional observer
satisfying \eqref{eq:observer-embedding} must satisfy
\begin{equation}
    \mu_i^{\rm rem}\mid\beta_i.
    \label{eq:factor-allocation-condition}
\end{equation}
Conversely, if $\nu_i\geq\kappa_C$, every real monic polynomial
$\beta_i$ of degree $\nu_i$ satisfying
\eqref{eq:factor-allocation-condition} is realizable by an
order-$\nu_i$ functional observer satisfying
\eqref{eq:observer-embedding}.
\end{theorem}

\begin{proof}
By Proposition~\ref{prop:functional-observer-existence},
\(
    \ell_i\beta_i(A)
    \in
    \operatorname{row}\mathcal O_{\nu_i}(C,A)
    \subseteq\mathcal U(C)^\perp.
\)
Hence, $\ell_i\beta_i(A)V_{\mathcal U}=0$. Using
$\ell_i=f_i^\top\phi_i(A)$ and
\eqref{eq:unobservable-dynamics} gives
\(
    f_i^\top V_{\mathcal U}
    \phi_i(A_{\mathcal U})\beta_i(A_{\mathcal U})=0.
\)
Equation~\eqref{eq:residual-polynomial-property} then gives
$\mu_i^{\rm rem}\mid\beta_i$.

Conversely, if $\mu_i^{\rm rem}\mid\beta_i$, then
\eqref{eq:residual-polynomial-property} gives
$\ell_i\beta_i(A)V_{\mathcal U}=0$, and hence
$\ell_i\beta_i(A)\in\mathcal U(C)^\perp$. For
$\nu_i\geq\kappa_C$,
\(
    \mathcal U(C)^\perp
    =
    \operatorname{row}\mathcal O_{\nu_i}(C,A),
\)
so Proposition~\ref{prop:functional-observer-existence} gives the
desired observer.
\end{proof}

The residual polynomial characterizes both fixed-gain and
constraint-level HOCBF observability and detectability.

\begin{corollary}\longthmtitle{HOCBF observability and detectability}
\label{cor:hocbf-observability}
\leavevmode\par
\begin{enumerate}
    \item For fixed $\boldsymbol{\alpha}_i$,
    $(h_i,\boldsymbol{\alpha}_i)$ is HOCBF-observable if and only if
    $\mu_i^{\rm rem}=1$ (equivalently, $\mu_i\mid\phi_i$), and
    HOCBF-detectable if and only if $\mu_i^{\rm rem}$ is Hurwitz
    (equivalently, $\mu_i$ is Hurwitz).
    \item $h_i$ is HOCBF-observable if and only if all roots of $\mu_i$ are real
    and strictly negative with $\deg(\mu_i)\leq r_i$, and
    HOCBF-detectable if and only if $\mu_i$ is Hurwitz.
\end{enumerate}
\end{corollary}

\begin{proof}
For fixed $\boldsymbol{\alpha}_i$, Definition~
\ref{def:fixed-hocbf-observability} and
\eqref{eq:residual-polynomial-property} give observability if and only
if $\mu_i^{\rm rem}=1$. Detectability holds if and only if the
remaining unobservable modes decay, i.e.,
$\mu_i^{\rm rem}$ is Hurwitz. Since $\phi_i$ has only real strictly
negative roots, removing common factors with $\phi_i$ cannot remove a
non-Hurwitz root; hence $\mu_i^{\rm rem}$ is Hurwitz if and only if
$\mu_i$ is Hurwitz.

By Definition~\ref{def:hocbf-observable-constraint}, $h_i$ is HOCBF-observable if and only if some admissible $\phi_i$ contains $\mu_i$. Since $\phi_i$ has $r_i$ real strictly negative roots, this holds exactly when all roots of $\mu_i$ are real and strictly negative and $\deg(\mu_i)\leq r_i$. Detectability follows from its gain independence.
\end{proof}

Thus, HOCBF observability requires $\mu_i$ to have only real negative
roots with $\deg(\mu_i)\leq r_i$, whereas detectability is
gain-independent and requires only $\mu_i$ to be Hurwitz.

\begin{algorithm}[t]
\caption{HOCBF--observer co-design}
\label{alg:hocbf-observer-codesign}
\begin{algorithmic}[1]
\State Compute $c_i=f_i^\top V_{\mathcal U}$ and obtain $\mu_i$
from \eqref{eq:krylov-dependence}.

\State Select up to $r_i$ real negative roots $\zeta$ of $\mu_i$,
counting multiplicity, and set $\alpha_{ij}=-\zeta$. To achieve
HOCBF observability, assign all factors when permitted by
Corollary~\ref{cor:hocbf-observability}. Choose any remaining gains
positive.

\State Compute
$\mu_i^{\rm rem}=\mu_i/\gcd(\mu_i,\phi_i)$.

%\State Choose $\nu_i=\max\{\deg(\mu_i^{\rm rem}),\kappa_C\}$ and a real monic $\beta_i$ such that $\mu_i^{\rm rem}\mid\beta_i$. If vanishing estimation error is required, $\mu_i$ must be Hurwitz; choose any remaining roots of $\beta_i$ as desired, e.g., to meet a specified convergence rate.

\State Set
$\nu_i=\max\{\deg(\mu_i^{\rm rem}),\kappa_C\}$ and choose a real monic
$\beta_i$ with $\mu_i^{\rm rem}\mid\beta_i$. For vanishing error,
require $\mu_i$ Hurwitz and assign the remaining roots to meet the
desired convergence rate.

\State Construct the observer matrices satisfying
\eqref{eq:observer-embedding} using a standard functional-observer
synthesis method \cite{MD:00,KV:21}.
\end{algorithmic}
\end{algorithm}

Theorem~\ref{thm:factor-allocation} gives
\(
    \nu_i\geq\deg(\mu_i^{\rm rem}),
\)
while
\(
    \nu_i=\max\{\deg(\mu_i^{\rm rem}),\kappa_C\}
\)
is always sufficient and is therefore a natural default. Larger orders
add freely assignable factors to $\beta_i$ at the cost of a
higher-dimensional observer, without removing the mandatory factor
$\mu_i^{\rm rem}$.

Algorithm~\ref{alg:hocbf-observer-codesign} summarizes the
co-design procedure.

\subsection{Estimation-Rate Allocation}
\label{subsec:estimation-rate-allocation}

Fix $\boldsymbol{\alpha}_i$ and suppose
$(h_i,\boldsymbol{\alpha}_i)$ is HOCBF-detectable.
By Corollary~\ref{cor:hocbf-observability},
$\mu_i^{\rm rem}$ is Hurwitz. Theorem~\ref{thm:factor-allocation}
requires every root of $\mu_i^{\rm rem}$ to be an observer-error pole,
thereby limiting the achievable estimation rate. Define
\[
    \Gamma_i
    :=
    \begin{cases}
        -\displaystyle
        \max_{\mu_i^{\rm rem}(\zeta)=0}
        \operatorname{Re}(\zeta),
        &\mu_i^{\rm rem}\neq1,\\[3mm]
        +\infty,
        &\mu_i^{\rm rem}=1.
    \end{cases}
\]
Thus, $\Gamma_i>0$ when $\mu_i^{\rm rem}\neq1$, while
$\mu_i^{\rm rem}=1$ imposes no mandatory observer-error pole.

An observer achieves exponential rate $\lambda$ if
\(
    |e_i(t)|\leq Me^{-\lambda t}
\)
for some $M>0$ over the considered initial errors.

\begin{theorem}\longthmtitle{Estimation-rate limit}
\label{thm:estimation-rate-limit}
Fix $\boldsymbol{\alpha}_i$ and suppose
$(h_i,\boldsymbol{\alpha}_i)$ is HOCBF-detectable. Every rate
$0\leq\lambda<\Gamma_i$ is achievable by a scalar functional observer
of order
\[
    \nu_i=
    \max\{\deg(\mu_i^{\rm rem}),\kappa_C\}.
\]
No  functional observer satisfying
\eqref{eq:observer-embedding} guarantees
$\lambda>\Gamma_i$.
\end{theorem}

\begin{proof}
Let $\lambda<\Gamma_i$. Every root of $\mu_i^{\rm rem}$ satisfies
$\operatorname{Re}(\zeta)<-\lambda$. Complete
$\mu_i^{\rm rem}$ to a real monic polynomial $\beta_i$ of degree
$\max\{\deg(\mu_i^{\rm rem}),\kappa_C\}$ by adding roots with real part more negative than $-\lambda$. Theorem~\ref{thm:factor-allocation} then gives an
observer whose error poles satisfy
$\operatorname{Re}(\zeta)<-\lambda$.
Conversely, if $\mu_i^{\rm rem}\neq1$,
Theorem~\ref{thm:factor-allocation} requires
$\mu_i^{\rm rem}\mid\beta_i$. Thus, every admissible error polynomial
contains a root with real part $-\Gamma_i$. By minimality of the error
realization, this mode affects $e_i$ for some initial error, precluding
any guaranteed rate greater than $\Gamma_i$. For
$\mu_i^{\rm rem}=1$, $\Gamma_i=+\infty$ and the upper bound is
vacuous.
\end{proof}

\subsection{Set-Maximizing Ordering for Fixed Gains}
\label{subsec:gain-ordering}

After the gain values are selected, their ordering remains free.
Permuting them preserves $\phi_i$, the HOCBF functional, observability
properties, and observer design, but may change the HOCBF admissible set.

Let $\boldsymbol{\alpha}_i^\downarrow$ and
$\boldsymbol{\alpha}_i^\uparrow$ denote the gains in nonincreasing and
nondecreasing order, respectively, and
$\boldsymbol{\alpha}_i^\pi$ any other ordering. Denote the
corresponding collections over all constraints by
$\boldsymbol{\alpha}^\downarrow$,
$\boldsymbol{\alpha}^\uparrow$, and
$\boldsymbol{\alpha}^\pi$.

\begin{theorem}\longthmtitle{Set-maximizing order for fixed gains}
\label{thm:set-maximizing-order}
Let $\mathcal C_{\boldsymbol{\alpha}}$ be defined by
\eqref{eq:hocbf-set}. For any ordering $\boldsymbol{\alpha}^\pi$ of
the selected gains,
\begin{equation}
    \mathcal C_{\boldsymbol{\alpha}^\uparrow}
    \subseteq
    \mathcal C_{\boldsymbol{\alpha}^\pi}
    \subseteq
    \mathcal C_{\boldsymbol{\alpha}^\downarrow}.
    \label{eq:gain-order-set-inclusion}
\end{equation}
\end{theorem}

\begin{proof}
Fix constraint $i$ and two orderings differing only by adjacent gains
$b\leq a$ at positions $k$ and $k+1$. Let $\psi_j$ and
$\widetilde\psi_j$ correspond to $(b,a)$ and $(a,b)$, respectively.
They coincide for $j<k$, while
\(
    \widetilde\psi_k-\psi_k=(a-b)\psi_{k-1}.
\)
Thus, $\psi_{k-1}(x),\psi_k(x)\geq0$ imply
$\widetilde\psi_k(x)\geq0$. At level $k+1$, the two recursions
coincide because
\(
\left(\frac{d}{dt}+a\right)
\left(\frac{d}{dt}+b\right)
=
\left(\frac{d}{dt}+b\right)
\left(\frac{d}{dt}+a\right).
\)
Since all subsequent gains are identical, all later HOCBF functions
also coincide. Thus, moving a larger gain earlier cannot decrease the
admissible set.
Applying this argument to each constraint, if
$x\in\mathcal C_{\boldsymbol{\alpha}^\uparrow}$, repeatedly moving
larger gains earlier yields any $\boldsymbol{\alpha}^\pi$, so
$x\in\mathcal C_{\boldsymbol{\alpha}^\pi}$. Likewise, if
$x\in\mathcal C_{\boldsymbol{\alpha}^\pi}$, the same swaps yield
$\boldsymbol{\alpha}^\downarrow$, so
$x\in\mathcal C_{\boldsymbol{\alpha}^\downarrow}$. Hence,
\eqref{eq:gain-order-set-inclusion} holds.
\end{proof}

For an HOCBF-observable constraint with
$\deg(\mu_i)=r_i$, observability fixes $\phi_i=\mu_i$, leaving only
the ordering to optimize. Otherwise, free gain values remain; optimizing
them changes $\phi_i$ and the observer design and may require additional
gain bounds to admit a finite maximizer.

\section{Aircraft Roll--Yaw Example}
\label{sec:aircraft-example}

%This section illustrates HOCBF observability, factor allocation, observer order and rate, and gain ordering on a linearized aircraft roll--yaw model with an unobservable full state.

This section illustrates the co-design results on a linearized aircraft
roll--yaw model with unobservable full state.

Consider the linearized roll--yaw dynamics of a mid-size aircraft at
the operating condition in \cite[Sec.~14.8]{Lavretsky-Wise:24}, with
velocity $717.17~\mathrm{ft/s}$, altitude $25{,}000~\mathrm{ft}$, and
angle of attack $4.5627^\circ$. Using the aileron channel,
\begin{equation*}
\begin{aligned}
    \dot x&=Ax+Bu,\qquad
    x=\begin{bmatrix}\beta_{\rm s}&p_{\rm s}&r_{\rm s}\end{bmatrix}^{\top},\\
A&=
{\footnotesize
\begin{bmatrix}
-0.11794&0.000848&-1.0001\\
-7.0113&-1.4492&0.22059\\
6.3035&0.065114&-0.41172
\end{bmatrix}},\quad
B=
{\footnotesize
\begin{bmatrix}
0\\-7.9662\\0.60926
\end{bmatrix}}.
\end{aligned}
\end{equation*}
Here, $\beta_{\rm s}$ is the sideslip angle and $p_{\rm s}$ and
$r_{\rm s}$ are the stability-axis roll and yaw rates.

The matrix $A$ has a stable real eigenvalue
$\lambda_{\rm u}\approx-1.3902$. Let $v_{\rm u}$ be a corresponding
eigenvector and choose $C$ such that
$\ker(C)=\operatorname{span}\{v_{\rm u}\}$. Numerically,
\begin{equation*}
C\simeq
{\footnotesize
\begin{bmatrix}
 0.99991&0.00884&0.01018\\
-0.01026&0.01018&0.99990
\end{bmatrix}}.
\end{equation*}
Thus, $\operatorname{rank}\mathcal O_2(C,A)=2<3$ and
$\kappa_C=0$, so the full state is unobservable. Reported values are
rounded, while computations use full precision.

We impose the two-sided sideslip constraint
\[
    |\beta_{\rm s}|\leq\beta_{\max},
    \qquad
    \beta_{\max}=0.25^\circ.
\]
Both constraints have relative degree two, with relevant unobservable
polynomial
\(
    \mu(s)=s-\lambda_{\rm u}.
\)

For Design~1, choose
\(
    \phi^{(1)}(s)=(s+2)(s+4).
\)
Since $\phi^{(1)}(\lambda_{\rm u})\neq0$,
$\mu^{\rm rem}(s)=s-\lambda_{\rm u}$. Hence, for these gains the
HOCBF is detectable but not observable. Among observers satisfying
\eqref{eq:observer-embedding}, order one is necessary and sufficient,
with unavoidable error pole $\lambda_{\rm u}$ and rate limit
$\Gamma=-\lambda_{\rm u}\approx1.3902$. One realization for the upper
constraint is
\begin{equation*}
\begin{aligned}
    \dot\eta
        &\simeq\lambda_{\rm u}\eta
        +\begin{bmatrix}-0.0985&0.0041\end{bmatrix}y
        -0.1108u,\\
    \hat z
        &\simeq\eta+
        \begin{bmatrix}-0.9399&5.4810\end{bmatrix}y .
\end{aligned}
\end{equation*}

%For Design~2, choose 
%\(
 %   \phi^{(2)}(s)=(s-\lambda_{\rm u})(s+4).
%\)
%The hidden factor is now assigned to the HOCBF polynomial, so $\mu^{\rm rem}=1$ and the HOCBF is observable. Hence, the sideslip constraint is HOCBF-observable. Since $\kappa_C=0$,
%\(
 %   \hat z\simeq
    %\begin{bmatrix}1.4209&4.8469\end{bmatrix}y,
%\)
%and no dynamic observer state is required. Thus, within the considered
%observer architecture, the required dynamic order drops from one to zero, illustrating Theorem~\ref{thm:factor-allocation}.

For Design~2, choose
\(
    \phi^{(2)}(s)=(s-\lambda_{\rm u})(s+4).
\)
Assigning the hidden factor to $\phi^{(2)}$ gives
$\mu^{\rm rem}=1$, making the sideslip constraint HOCBF-observable.
Since $\kappa_C=0$,
\(
    \hat z\simeq
    \begin{bmatrix}1.4209&4.8469\end{bmatrix}y,
\)
so no dynamic observer state is needed and the observer order drops
from one to zero, illustrating Theorem~\ref{thm:factor-allocation}.

For the simulations, $x(0)=0$ and
\[
    u_{\rm nom}(t)
    =-0.11[\tanh(10(t-1))-\tanh(10(t-6))]~\mathrm{rad},
\]
which violates the sideslip limit. Figure~\ref{fig:aircraft-results}(a) shows a nominal peak of about
$1.45^\circ$, while both filters maintain
$|\beta_{\rm s}|\leq0.25^\circ$, confirming
Proposition~\ref{prop:one-sided-safety}. For Design~1,
$e(0)=-0.008$ and $\rho(t)=0.01e^{\lambda_{\rm u}t}$; the same $\rho$
applies to both constraints. Figure~\ref{fig:aircraft-results}(b) confirms the bound in
Proposition~\ref{prop:componentwise-error-bound} and the rate predicted
by Theorem~\ref{thm:estimation-rate-limit}. The QP remains feasible
although $D$ is not full row rank.

%Finally, permuting the gains $\{-\lambda_{\rm u},4\}$ in Design~2 leaves the final HOCBF inequality and observer unchanged. The intermediate conditions are
For Design~2, gain order does not affect the HOCBF inequality or
observer; the intermediate conditions are
\[
    -\alpha_1(\beta_{\max}+\beta_{\rm s})
    \leq\dot\beta_{\rm s}
    \leq
    \alpha_1(\beta_{\max}-\beta_{\rm s}).
\]
The nondecreasing order has $\alpha_1=-\lambda_{\rm u}$, whereas the
nonincreasing order has $\alpha_1=4$.
Figure~\ref{fig:aircraft-results}(c) shows the corresponding
admissible regions in the $(\beta_{\rm s},\dot\beta_{\rm s})$ plane:
the region for $\boldsymbol{\alpha}^{\downarrow}$ contains that for
$\boldsymbol{\alpha}^{\uparrow}$, as predicted by
Theorem~\ref{thm:set-maximizing-order}.

\begin{figure}[t]
    \centering
    \includegraphics[width=\columnwidth]{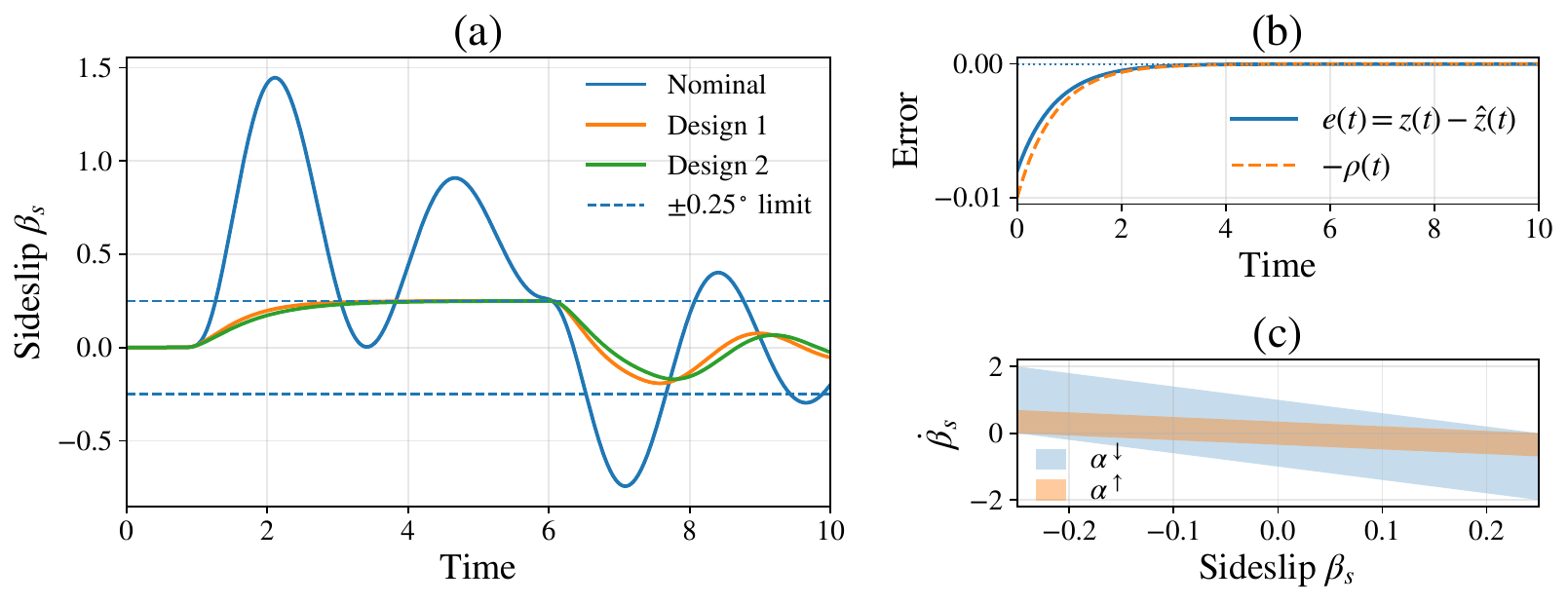}
    \caption{Aircraft roll--yaw simulation: (a) sideslip response;
    (b) estimation error and bound; (c) HOCBF admissible
    regions for the two gain orderings.}
    \label{fig:aircraft-results}
\end{figure}

%%%%%%%%%%%%%%%%%%%%%%%%%%%%%%%%%%%%%%%%%%%%%%%%%%%%%%%%%%%%%%%%%%%%%%%
\section{Conclusion}
\label{sec:conclusion}

This letter developed an output-feedback HOCBF safety-filtering
framework based on scalar functional observers, avoiding unnecessary
full-state reconstruction. A polynomial co-design characterizes
fixed-gain HOCBF observability and detectability, their achievability
through gain selection, and observer-order and rate limits. Certified
one-sided error bounds ensure safety, while gain ordering enlarges the
HOCBF admissible set. Future work will consider nonlinear systems and
shared functional observers.

%%%%%%%%%%%%%%%%%%%%%%%%%%%%%%%%%%%%%%%%%%%%%%%%%%%%%%%%%%%%%%%%%%%%%%%

\bibliographystyle{IEEEtran}
\bibliography{reference}

\end{document}